\documentclass[a4paper,9pt]{article}
\usepackage[russian,english]{babel}
\usepackage{longtable}
\usepackage{amsmath}
\usepackage{amssymb}
\usepackage{amsthm}
\usepackage{listings}
\usepackage{graphicx}
\usepackage[pdftitle={Constant-time connectivity tests},pdfauthor={Dr.\ Philipp Klaus Krause}]{hyperref}

\newcommand{\coloneqq}{\mathrel{\mathop:}=}

\newtheorem{lemma}{Lemma}
\newtheorem{theorem}{Theorem}

\author{Philipp Klaus Krause}
\title{Constant-time connectivity tests}

\begin{document}

\maketitle

\begin{abstract}
We present implementations of constant-time algorithms for connectivity tests and related problems. Some are implementations of slightly improved variants of previously known algorithms; for other problems we present new algorithms with substantially better runtime than previous algorithms: estimates of the distance to and tolerant testers for connectivity, $2$-edge-connectivity, $3$-edge-connectivity, eulerianity.
\end{abstract}

\section{Introduction}

Property testing is concerned with extremely fast (constant-time or other sublinear) algorithms for approximate decision-making.
While the runtime of constant-time algorithms does not depend on the size of the input, the runtime depending on parameters, such as the average degree of the input graph, or on the maximum allowed error of the output are still the subject of research.


We implemented constant-time graph algorithms for testing connectivity, $2$-edge connectivity and $3$-edge connectivity for sparse graphs.
We implemented an estimate of the number of connected components, estimates of the distance to connectivity, $2$-edge-connectivity and eulerianity, tolerant testers for connectivity, $2$-edge-connectivity and eulerianity.
We also present an algorithm for estimating the distance to $3$-edge-connectivity and a tolerant tester for $3$-edge-connectivity that have not been implemented yet.
For connectivity and eulerianity our approach has better worst-case runtime than previous approaches (and the same expected runtime). For $2$-edge-connectivity and $3$-edge-connectivity, our approaches have better runtime than the previously known algorithm for general $k$-edge-connectivity. These advantages carry over to the tolerant testers.

Full, compileable C source of our implementation can be found as free software~\cite{Stallman2002} at \url{http://zshg.sourceforge.net}.

\section{Preliminaries}

Let $G = (V, E)$ be a graph. For a set of nodes $U \subseteq V$, the degree $d(U)$ is the number of edges with exactly one endpoint in $U$. For a single node $v \in V$ we also define the degree $d(v) \coloneqq d(\{v\})$.
The number $\Delta(G) \coloneqq \max\{d(v)\ |\ v \in V\}$ is the \emph{maximum degree} of $G$.
The number
\begin{displaymath}
d(G) \coloneqq \frac{1}{|V|}\sum_{v \in V}d(v) = 2\frac{|E|}{|V|}
\end{displaymath}
is the \emph{average degree} of $G$.

In the \emph{bounded-degree model}, a graph $G = (V, E), \Delta \geq \Delta(G)$ is \emph{$\epsilon$-far} from having a property $\mathcal{P}$, if it cannot be transformed into a graph $G' \in \mathcal{P}, \Delta(G') \leq \Delta$ by at most $\epsilon \Delta |V|$ edge modifications. Complexities are given as functions of the maximum degree $\Delta$.

In the \emph{sparse graph model} (also called \emph{unbounded-degree model}), a graph $G = (V, E)$ is \emph{$\epsilon$-far} from having a property $\mathcal{P}$, if it cannot be transformed into a graph $G' \in \mathcal{P}$ by at most $\epsilon |E|$ edge modifications (i.e. edge insertions and edge deletions). Complexities are usually given as functions of the average degree $d$. When estimating distance to a property, often the distance is given in terms of edge modifications relative to the number of nodes ($\delta$): $\delta |V| = \epsilon |E|$.

The \emph{incidence-lists model} consists of the bounded-degree model and the unbounded-degree model. It corresponds to an implementation of the graph data structure using adjacency lists or incidence lists.

In the incidence lists model the available queries are \emph{degree} and \emph{$i$-th neighbour}.

Other models include the \emph{dense graph model} and the \emph{combined model} (also called the \emph{general graph model}). The dense graph model corresponds to an implementation of the graph data structure using an adjacency matrix. The available query is \emph{adjacency}. The combined model allows all queries from the incidence lists and the dense graph models.

All models allow uniform random sampling of nodes.

A graph is \emph{connected} if it has at least one node and for any two nodes $u,v$ in the graph, there is a path from $u$ to $v$. A \emph{connected component} of a graph is a maximal connected subgraph. A graph is called $k$\emph{-edge-connected}, if any graph obtained from it by removing less than $k$ edges is connected. An edge in a graph is a \emph{bridge}, if all paths connecting its endpoints contain the edge.
A $k$-class in a graph $g$ is a $k$-edge connected subset of nodes of a graph  $G$ (note that for $k > 2$ a $k$-class in $G$ is not necessarily a $k$-edge-connected subgraph of $G$).
A $k$-set in a graph is a $k$-edge-connected subgraph that is connected to the rest of the graph by less than $k$ edges. Every node in a graph is contained in at least one $k$-set. The $1$-sets are the connected components. A minimal $k$-set is called a $k$-class-leaf (the terminology is inspired by the use of the word ``leaf'' for $2$-class-leaves~\cite{Brahana1917} and by the creation of a auxiliary graph, that has one node per $k$-class in the original graph. For $k = 2$, this auxiliary graph is the bridge tree. The $k$-class-leaves correspond to nodes of degree $1$ in the auxiliary graph.

A set of nodes $U \subseteq V$ is \emph{$\ell$-extreme}, if $d(U) = \ell$ and each subset of $U$ has degree larger than $\ell$. Every $\ell$-extreme set is an $\ell'$-class for some $\ell'$ and any two extreme sets in a graph are either disjoint or one is a subset of the other~\cite{Naor1997}. For $U, W \subseteq V$ we use the notation $U \sqsubset W$ to denote that $U \subseteq W$ and $U$ is extreme and there is no extreme set $U'$, such that $U \subseteq U' \subseteq W$.

A graph is \emph{eulerian} if it contains a cycle that traverses each edge of the graph exactly once.

\section{Related Work}\label{Related}

\begin{figure}[!h]
\begin{longtable}{r|r|r}
	Property & Model & Complexity \\
	\hline
	Connectivity~\cite{Goldreich2002} & $\Delta$ & $O\left(\frac{\log^2\left(\frac{1}{\epsilon \Delta}\right)}{\epsilon}\right)$ \\
	\hline
	Connectivity~\cite{Ron2010} & $d$ & $O\left(\frac{\log\left(\frac{1}{\epsilon d}\right)}{\epsilon^2d}\right)$ \\
	\hline
	Connectivity~\cite{Parnas2002} & $d$ & $O\left(\frac{\log\left(\frac{1}{\epsilon d}\right)}{\epsilon^2d^2}\right)$ \\
	\hline
	Connectivity (here) & $d$ & $O\left(-\log(1 - p) \cdot \frac{\log\left(\frac{1}{\epsilon d}\right)}{\epsilon^2d^2}\right)$ \\
	\hline
	$2$-edge-connectivity~\cite{Goldreich2002} & $\Delta$ & $O\left(\frac{\log^2\left(\frac{1}{\epsilon \Delta}\right)}{\epsilon}\right)$ \\
	\hline
	$2$-edge-connectivity (here) & $d$ & $O\left(-\log(1 - p) \cdot \frac{\log\left(\frac{1}{\epsilon d}\right)}{\epsilon^2d^2}\right)$ \\
	\hline
	$3$-edge-connectivity~\cite{Goldreich2002} & $\Delta$ & $O\left(\frac{\log\left(\frac{1}{\epsilon \Delta}\right)}{\epsilon^2}\right)$ \\
	\hline
	$3$-edge-connectivity (here) & $d$ & $O\left(-\log(1 - p) \cdot \frac{\log\left(\frac{1}{\epsilon d}\right)}{\epsilon^3d^3}\right)$ \\
	\hline
	$k$-edge-connectivity~\cite{Goldreich2002} & $\Delta$ & $O\left(\frac{k^3 \log\left(\frac{1}{\epsilon \Delta}\right)}{\epsilon^{3 - \frac{2}{k}}\Delta^{2 - \frac{2}{k}}}\right)$ \\
	\hline
	$k$-edge-connectivity~\cite{Parnas2002} & $d$ & $\tilde{O}\left(\frac{k^4}{\epsilon^4d^4}\right)$ \\
	\hline
	Eulerianity~\cite{Goldreich2002} & $\Delta$ & $O\left(\frac{\log^2\left(\frac{1}{\epsilon \Delta}\right)}{\epsilon}\right)$ \\
	\hline
	Eulerianity~\cite{Parnas2002} & $d$ & $O\left(\frac{\log\left(\frac{1}{\epsilon d}\right)}{\epsilon^2d^2}\right)$ \\
	\hline
	Eulerianity (here) & $d$ & $O\left(-\log(1 - p) \cdot \frac{\log\left(\frac{1}{\epsilon d}\right)}{\epsilon^2d^2}\right)$ \\
\end{longtable}
\caption{\label{Connectivitytestcomparison}Comparison of connectivity tests, $\Delta$: bounded-degree model, $d$: sparse graph model.}
\end{figure}

\begin{figure}[!h]
\begin{longtable}{r|r|r}
	Property & Model & Complexity \\
	\hline
	Connectivity~\cite{Chazelle2005} & $d$ & $O\left(\frac{d \log\left(\frac{d}{\delta}\right)}{\delta^2}\right)$ \\
	\hline
	Connectivity~\cite{Marko2005} & $d$ & $O\left(\frac{1}{\epsilon^4d^4}\right) = O\left(\frac{1}{\delta^4}\right)$ \\
	\hline
	Connectivity~\cite{Berenbrink2014} & $d$ & $O\left(\frac{1}{\delta^4}\right)$, exp.: $O\left(\frac{\log\left(\frac{1}{\delta}\right)}{\delta^2}\right)$ \\
	\hline
	Connectivity (here) & $d$ & $O\left(\frac{\log\left(\frac{1}{\delta}\right)}{(1 - p)\delta^2}\right)$ \\
	\hline
	$2$-edge-connectivity (here) & $d$ & $O\left(\frac{\log\left(\frac{1}{\delta}\right)}{(1 - p)\delta^2}\right)$ \\
	\hline
	$3$-edge-connectivity (here) & $d$ & $O\left(\frac{\log\left(\frac{1}{\delta}\right)}{(1 - p)\delta^3}\right)$ \\
	\hline
	$k$-edge-connectivity~\cite{Marko2005} & $d$ & $O\left(\frac{k^6\log\left(\frac{k}{\epsilon d}\right)}{\epsilon^6d^6}\right) = O\left(\frac{k^6\log\left(\frac{k}{\delta}\right)}{\delta^6}\right)$ \\
	\hline
	Eulerianity~\cite{Marko2005} & $d$ & $O\left(\frac{1}{\epsilon^4d^4}\right) = O\left(\frac{1}{\delta^4}\right)$ \\
	\hline
	Eulerianity (here) & $d$ & $O\left(\frac{\log\left(\frac{1}{\delta}\right)}{(1 - p)\delta^2}\right)$ \\
\end{longtable}
\caption{\label{Connectivityestimatecomparison}Comparison of estimates of distance to connectivity, $d$: sparse graph model.}
\end{figure}

For connectivity, there is a relatively simple algorithm with small constants~\cite{Goldreich2002} in the bounded-degree model, that also works well in the sparse graph model~\cite{Ron2010}, where it has slightly higher complexity. With a more careful analysis, one can achieve a slightly better complexity bound~\cite{Parnas2002}.
Based on the existing approaches, we came up with a connectivity test that uses fewer queries (but has the same asymptotic complexity).

The algorithm for connectivity testing can be generalized to $k$-edge-connectivity testing~\cite{Goldreich2002}. These generalizations are more involved than the algorithm for connectivity, but still implementable, which also holds for their sparse graph versions~\cite{Parnas2002}.
A graph is eulerian when it is connected and has no nodes of odd degree~\cite{Euler1741}. Algorithms for connectivity often can be modified to get algorithms testing for eulerianity.

Tolerant testers can be created from algorithms that estimate the distance to having the property (as we do in Section~\ref{Tolerant}). There are various approaches to estimating the distance to connectivity. A recent one~\cite{Berenbrink2014} has an expected runtime of $O\left(\frac{\log\left(\frac{1}{\delta}\right)}{\delta^2}\right)$, though the worst-case complexity is substantially higher at $O\left(\frac{1}{\delta^4}\right)$. An earlier approach~\cite{Marko2005} can estimate the distance to $k$-edge-connectivity for any $k$. For $k = 2$ and $k = 3$ we can estimate the distance with lower asymptotic complexity than the general approach.

There was an attempt to construct tolerant testers for $k$-edge-connectivity via local reconstructors~\cite{Campagna2013}. The claimed complexity was substantially higher than previous approaches and exponential in $k$; there are multiple flaws in the proof, that cannot be fixed easily.

\section{Connectivity Test}\label{Connectivitytest}

We implemented a connectivity (and $2$- and $3$-edge connectivity) test for the sparse graph model.
It is similar to existing connectivity test algorithms~\cite{Goldreich2002,Ron2010,Parnas2002}, which also carries over to the proofs. Our approach (Figure~\ref{Connecitvitytestalgorithm}) has a few minor changes to reduce the query complexity by a constant factor. Also, instead of having a fixed probability of $\frac{2}{3}$ for rejecting graphs that are $\epsilon$-far from being connected, we allow a parameter $p$; compared to just running the algorithm repeatedly, this lowers runtime and number of queries when the acceptable probability $1 - p$ of an unconnected graph that is $\epsilon$-far from being connected being considered connected is not a power of $\frac{1}{3}$. The implementations of the algorithms rand\_index(), which uniformly selects a random node, and k\_set($s$, $2^i$), which return the number of nodes in the minimal $k$-set containing $s$ connected to the rest of the graph by less than $k$ edges are not shown. 
1\_set is a connectivity test (which can be implemented as a depth-first search). 2\_set() and 3\_set can be implemented using variants of algorithms for identifying 2-class-leaves and 3-class leaves~\cite{Goldreich2002}.

\begin{figure}[!h]
\lstset{mathescape=true}
\begin{lstlisting}
bool zshg_k($\epsilon$, $p$, $d$, $n$)
{
  $q \coloneqq -\log_e(1 - p)$;
	
  if($n$ == $0$) // Empty graph: unconnected
    return(false);
  if($d < 1$) // Low degree: connected if single node only
    return($n$ == 1);

  $\ell \coloneqq \log_2\left(\frac{8\lceil\frac{k}{2}\rceil}{\epsilon d}\right)$;

  if($n \leq \frac{8q\lceil\frac{k}{2}\rceil\ell}{\epsilon d}$) // Small graph: do exact check
    return(k_set(0, $n$) ==  $n$);

  for($i \coloneqq 1$; $i \leq \ell$; $i \coloneqq i + 1$)
  {
    $m_i \coloneqq \frac{8q\lceil\frac{k}{2}\rceil\ell}{2^i\epsilon d}$;

    for($c \coloneqq 0$; $c \leq m_i$; $c \coloneqq c + 1$)
    {
      $s$ $\coloneqq$ rand_index(n);
			
      if(k_set($s$, $2^i$) < $2^i$)
        return(false);
    }
  }

  return(true);
}
\end{lstlisting}
\caption{\label{Connecitvitytestalgorithm}$k$-edge-connectivity test algorithm in C-like pseudocode (the implementation resides in zshg\_impl() in zshg.c)}
\end{figure}

\begin{lemma}\label{lem:connectivitycorrectness}
The algorithm in Figure~\ref{Connecitvitytestalgorithm} accepts all $k$-edge-connected graphs and rejects graphs that are $\epsilon$-far from being $k$-edge-connected with probability at least $p$.
\end{lemma}

\begin{proof}
Clearly, the algorithm never rejects a $k$-edge-connected graph. If a graph is $\epsilon$-far from being $k$-edge-connected, it has more than $\frac{\epsilon m}{\lceil\frac{k}{2}\rceil}$ $k$-class-leaves (otherwise one could add $\epsilon mk$ edges to connect all $k$-class-leaves by adding $\lceil\frac{k}{2}\rceil$ cycles through them, resulting in a $k$-edge-connected graph). Then the graph has at least $\frac{\epsilon m}{2\lceil\frac{k}{2}\rceil}$ $k$-class-leaves consisting of less than $\frac{4\lceil\frac{k}{2}\rceil}{\epsilon d}$ nodes each. Let $B_i$ be the set of all such $k$-class-leaves of size at least $2^{i - 1}$, but less than $2^i$. Let $\ell \coloneqq \log_2(\frac{8\lceil\frac{k}{2}\rceil}{\epsilon d})$. Since $\sum_{i = 1}^{\ell} |B_i| \geq \frac{\epsilon m}{2\lceil\frac{k}{2}\rceil}$, there exists an $i, 1 \leq i \leq \ell$, such that $|B_i| \geq \frac{\epsilon m}{2\lceil\frac{k}{2}\rceil\ell}$. The number of nodes residing in $k$-class-leaves of size at least $2^{i - 1}$, but less than $2^i$ is at least $2^{i - 1}|B_i| \geq \frac{2^{i - 2}\epsilon m}{\lceil\frac{k}{2}\rceil\ell}$. The probability of selecting such a node by uniform sampling is $\frac{2^{i - 1}|B_i|}{n} \geq \frac{2^{i - 2}\epsilon m}{n \lceil\frac{k}{2}\rceil\ell} = \frac{2^{i - 3}\epsilon d}{\lceil\frac{k}{2}\rceil\ell}$. The probability of not selecting any such node in iteration $i$ of the algorithm bounds the probability $p_{\textrm{fail}}$ of not rejecting an unconnected graph that is $\epsilon$-far from being connected. Using $1 + x \leq e^x$, which follows from Bernoulli's inequality~\cite{Slusius1668,Bernoulli1689}, we get

\begin{align*}
p_{\textrm{fail}} \leq \left(1 - \frac{2^{i - 3}\epsilon d}{\left\lceil\frac{k}{2}\right\rceil\ell} \right)^{m_i} \leq e^{-\left(\frac{2^{i - 3}\epsilon d}{\left\lceil\frac{k}{2}\right\rceil\ell} m_i \right)} = e^{-\left(\frac{2^{i - 3}\epsilon d}{k\ell} \frac{8q\left\lceil\frac{k}{2}\right\rceil\ell}{2^i \epsilon d} \right)} = \\
= e^{-q} = e^{\log_e(1 - p)} = 1 - p.
\end{align*}
\end{proof}

\begin{lemma}\label{lem:connectivitycomplexity}
For $k = 1, 2$ the query complexity of the algorithm in Figure~\ref{Connecitvitytestalgorithm} is $O\left(-\frac{\log(1 - p)\log\left(\frac{1}{\epsilon d}\right)}{\epsilon^2 d^2}\right)$. For $k = 3$ it is $O\left(-\frac{\log(1 - p)\log\left(\frac{1}{\epsilon d}\right)}{\epsilon^3 d^3}\right)$.
\end{lemma}

\begin{proof}
In each iteration of the algorithm, $m_i$ nodes are selected. For $k = 1$, for each selected node, in k\_set() a depth-first-search for up to $2^i$ nodes is done (dfs() in zshg.c), resulting in less than $2^i \cdot 2^i$ queries for each search. The query complexity is thus less than
\begin{align*}
\sum_{i = 1}^{\ell}m_i \cdot 2^i \cdot 2^i = \sum_{i = 1}^{\ell}\frac{8q\ell}{2^i\epsilon d} \cdot 2^i \cdot 2^i \leq \frac{8q\ell}{\epsilon d}2^{\ell + 1} = \frac{8q\log_2\left(\frac{8}{\epsilon d}\right)}{\epsilon d} 2 \frac{8}{\epsilon d} = \\
= 128 \frac{q \log_2\left(\frac{8}{\epsilon d}\right)}{\epsilon^2 d^2} = 128 \frac{-\log_e\left(1 - p\right) \cdot \log_2\left(\frac{8}{\epsilon d}\right)}{\epsilon^2 d^2} \in O\left(-\log\left(1 - p\right) \cdot \frac{\log\left(\frac{1}{\epsilon d}\right)}{\epsilon^2 d^2}\right).
\end{align*}

For $k = 2$, zshg uses a variant of an algorithm for 2-class-leaves~\cite{Goldreich2002} (comp2() in zshg2.c). This algorithm essentially does two depth-first-searches for up to $2^i$ nodes, with the second search somewhat restricted, for a total of at most $(2^i)^2$ queries. This results in a call to 2\_set needing up to $4$ times as many queries as 1\_set, and thus a total factor of about $4$ in the number of queries done by zshg\_2 vs. zshg\_1. For $k = 3$, we use a variant of an algorithm for 3-class-leaves~\cite{Goldreich2002}. This algorithm essentially does one depth-first-search for up to $2^i$ nodes followed by an invocation of 2\_set() on each node discovered. We get an additional factor of about $\frac{16}{\epsilon d}$ in the number of queries by zshg\_3 vs. zshg\_1.
\end{proof}

From our connectivity test, we obtain a test for eulerianity based on the following:

\begin{lemma}\label{lem:eulerianityconnectivity}
Let $G = (V, E)$ be a graph that is $\epsilon$-far from eulerianity. Then, it is $\frac{\epsilon}{2}$-far from connectivity or it has more than $\epsilon dn$ nodes of odd degree.
\end{lemma}

\begin{proof}
Assume that the graph is $\frac{\epsilon}{2}$-close to connectivity (i.e. it has at most $\frac{\epsilon}{2}dn$ connected components) and at most $\frac{\epsilon}{2}dn$ nodes of odd degree. We show that we can make the graph Eulerian by adding $\epsilon dn$ edges, a contradiction.
we first add $\frac{\epsilon}{2}dn$ edges connecting nodes of odd degree to obtain a graph in which all nodes have even degree. We then add $\frac{\epsilon}{2}dn$ edges, choosing one node out of each connected component in the original graph and adding a cycle through these. The resulting graph is connected and has only nodes of even degree and is thus eulerian~\cite{Euler1741,Listing1847,Hierholzer1873}.
\end{proof}

\begin{theorem}
There are testers for $1$- and $2$-edge-connectivity and for eulerianity with complexity $O\left(-\frac{\log(1 - p)\log\left(\frac{1}{\epsilon d}\right)}{\epsilon^2 d^2}\right)$, a tester for $3$-edge-connectivity with complexity $O\left(-\frac{\log(1 - p)\log\left(\frac{1}{\epsilon d}\right)}{\epsilon^3 d^3}\right)$.
\end{theorem}

\begin{proof}
This follows from Lemmata \ref{lem:connectivitycorrectness}, \ref{lem:connectivitycomplexity}, \ref{lem:eulerianityconnectivity}.
\end{proof}

\section{Distance to Connectivity Estimate}

Our algorithm for estimating the number of connected components is based on an earlier algorithm~\cite{Berenbrink2014}, which had worst-case query complexity $O(\delta^{-4})$ and expected query complexity $O(\delta^{-2}\log(\delta^{-1}))$. We reduce the worst-case query complexity to $O(\delta^{-2}\log(\delta^{-1}))$ and allow the probability $p$ of the error being outside the error bound to be specified instead of using a fixed value of $\frac{3}{4}$.

rand\_range($i$) returns a random integer in the range $1, \ldots, i$. Each number $j$ is returned with probability $j^{-2} - (j + 1)^{-2}$ for $j < i$, while $i$ is returned with probability $i^{-2}$. The function 1\_set() can be implemented as above in Section~\ref{Connectivitytest}; the additional parameter \&$\ell$ is used to limit the total number of queries made: $\ell$ is decremented each time a query is made; when $\ell$ reaches $0$, no further queries are made.

\begin{figure}[!h]
\lstset{mathescape=true}
\begin{lstlisting}
float 1_sets($\delta$, $p$, $n$)
{
  $r = \frac{2}{(1 - p)\delta^2}$;
  $a = 0$;

  $\ell = r\left(\log_e\left(\frac{2}{\delta}\right) + \frac{9}{2}\right)$;

  for($i$ = 0; $i < r$; $i$++)
  {
    $s$ = rand_index($n$);
    $x$ = rand_range($\frac{2}{\delta}$);
    $b$ = 1_set($s$, $x + 1$, &$\ell$);
    if($b \leq x$)
      $a$ += $b$;
  }

  return($\frac{an}{r} + \frac{\delta n}{4}$);
}
\end{lstlisting}
\caption{\label{Componentcountingalgorithm}$1$-set counting algorithm in C-like pseudocode (implementation: zshg\_components() in zshg\_c.c)}
\end{figure}

\begin{lemma}\label{lem:componentcorrectness}
For $\delta \leq 1$, with probability at least $p$, the return value of the algorithm in Figure~\ref{Componentcountingalgorithm} is the number of connected components in the graph up to an error of at most $\delta n$.
\end{lemma}

\begin{proof}
Let $X_i$ be the value of $x$ in iteration $i$. Let $c$ be the number of connected components. Let $c^*$ be the number of connected components of size at most $\frac{2}{\delta}$. Let $C$ be the return value of the algorithm in Figure~\ref{Componentcountingalgorithm}. Let $\hat{c}$ be the return value when ignoring the query bound $\ell$.
Let $B_i$ be the value of $b$ in iteration $i$ if $b < X_i$ and $0$ otherwise. Let $C_s$ be the size of the connected components containing the node $s$.

\begin{align*}
E(B_i) &= \frac{1}{n} \sum_{s \in V} C_s P(X_i \geq C_s) = \frac{1}{n} \sum_{\substack{s \in V \\ C_s \leq \frac{2}{\delta}}} \frac{1}{C_s} = \frac{c^*}{n} \leq 1.\\
E(B_i^2) &= \frac{1}{n} \sum_{s \in V} C_s^2 P(X_i \geq C_s) = \frac{1}{n} \sum_{\substack{s \in V \\ C_s \leq \frac{2}{\delta}}} 1 \leq 1.\\
Var(B_i) &= E(B_i^2) - E(B_i)^2 \leq E(B_i^2) \leq 1.\\
\hat{c} &= \frac{a}{n}\sum_{i = 0}^{r - 1} + \frac{\delta n}{4}. \\
E(\hat{c}) &= \frac{rE(B_i)n}{r} + \frac{\delta n}{4} = c^* + \frac{\delta n}{4}. &
Var(\hat{c}) &= \frac{n^2}{r}Var(B_i) \leq \frac{n^2}{r}.
\end{align*}

The number of connected components bigger than $\frac{2}{\delta}$ is at most $\frac{\delta n}{2}$. This allows us to bound the probability of the estimate being far off~\cite{TchEbychef1867} (still ignoring the query bound $\ell$).

\begin{align*}
c - \frac{\delta n}{4} < c^* + \frac{\delta n}{4} = E(\hat{c}) = c^* + \frac{\delta n}{4} \leq c + \frac{\delta n}{4}.\\
P\left(\left|\hat{c} - E(\hat{c})\right| > \frac{3\delta n}{4}\right) < \frac{Var(\hat{c})}{\left(\frac{3\delta n}{4}\right)^2} \leq \frac{8}{9}(1 - p).
\end{align*}

Let $Q_i$ be the number of queries made by a call to 1\_set($s$, $X_i$) for a uniformly chosen node $s$ (i.\,e. we consider the behaviour of the algorithm as if the limit $\ell$ wasn't there). The procedure 1\_set can be implemented as a depth-first-search, resulting in at most $\frac{x(x + 1)}{2}$ queries for 1\_set($s$, $x$).

\begin{align*}
E(Q_i) &\leq \sum_{j = 1}^{\frac{2}{\delta} - 1} \left(\frac{1}{j^2} - \frac{1}{(j + 1)^2}\right)\frac{j(j + 1)}{2} + \frac{1}{(\frac{2}{\delta})^2} \frac{\frac{2}{\delta}(\frac{2}{\delta} + 1)}{2} =\\
&= \sum_{j = 1}^{\frac{2}{\delta} - 1} \left(\frac{2j + 1}{2j(j + 1)}\right) + \frac{2 + \delta}{4} \leq \log_e\left(\frac{2}{\delta}\right) + \frac{1}{2}.\\
E(Q_i^2) &\leq \sum_{j = 1}^{\frac{2}{\delta} - 1} \left(\frac{1}{j^2} - \frac{1}{(j + 1)^2}\right)\left(\frac{j(j + 1)}{2}\right)^2 + \frac{1}{(\frac{2}{\delta})^2} \left(\frac{\frac{2}{\delta}(\frac{2}{\delta} + 1)}{2}\right)^2 =\\
&= \sum_{j = 1}^{\frac{2}{\delta} - 1} \left(\frac{(j + 1)^2}{4} - \frac{j^2}{4}\right) + \frac{(\frac{2}{\delta} + 1)^2}{4} =
= \frac{1}{2}\sum_{j = 1}^{\frac{2}{\delta} - 1}j + \frac{1}{4}\left(\frac{2}{\delta} - 1\right) + \frac{(\frac{2}{\delta} + 1)^2}{4} = \frac{2}{\delta^2} + \frac{1}{\delta}.\\
Var(Q_i) &= E(Q_i^2) - E(Q_i)^2 \leq E(Q_i^2) \leq \frac{2}{\delta^2} + \frac{1}{\delta}.
\end{align*}

This allows us to bound the probability of hitting the query bound $\ell$~\cite{TchEbychef1867}.

\begin{align*}
P\left(\sum_{i = 0}^{r - 1}Q_i > \ell\right) \leq P\left(\sum_{i = 0}^{r - 1}Q_i > r\left(E(Q_0) + 4\right)\right) < \frac{Var(Q_0)}{16r} < \frac{1}{9}(1 - p).
\end{align*}

We can thus bound the probability of the algorithm returning a result outside the error bounds.

\begin{align*}
P\left(|C - c| > \delta n\right) \leq P\left(\sum_{i = 0}^{r - 1}Q_i > \ell \textrm{ or } \left|\hat{c} - E(\hat{c})\right| > \frac{3\delta n}{4}\right) \leq\\
\leq P\left(\sum_{i = 0}^{r - 1}Q_i > \ell\right) + P\left(\left|\hat{c} - E(\hat{c})\right| > \frac{3\delta n}{4}\right) < (1 - p).
\end{align*}
\end{proof}

\begin{lemma}\label{lem:componentcomplexity}
The algorithm in Figure~\ref{Componentcountingalgorithm} has query complexity $O(\frac{1}{(1 - p)\delta^2} \log_e(\frac{1}{\delta}))$.
\end{lemma}

\begin{proof}
This follows directly from the choice of $\ell$.
\end{proof}

While our algorithm for estimating the distance to connectivity did directly reuse the function 1\_set from the connectivity test, this doesn't work for general $k$-edge-connectivity, since for $k \geq 2$, not every $k$-set is a $k$-class-leaf. For $k = 2$ we use:

\begin{lemma}[Eswaran and Tarjan~\cite{Eswaran1976}]\label{Lemma:Eswaran1976}
Let $G = (V, E)$ be a graph. Let $c_1$ be the number of 2-class-leaves that are 1-class-leaves in $G$. Let $c_2$ be the number of 2-class-leaves that are not 1-class-leaves in $G$. If $|V| > 2$ and $c_1 + c_2 > 1$, the minimal number of edge modifications necessary to make $G$ $2$-edge-connected is $\left\lceil \frac{c_2}{2} \right\rceil + c_1$.
\end{lemma}

Using classic terminology, $c_1$ is the number of connected components in $G$ that are $2$-edge-connected and $c_2$ is the number of $2$-edge-connected components in $G$ that are connected to the rest of the graph by a single bridge. To estimate the distance to $2$-edge-connectivity we want to estimate the numbers $c_1$ and $c_2$.

We use the algorithm in Figure~\ref{12leaf}, to find out if a node $s$ is in a $2$-edge-connected component of size up to $\mathfrak{r}$, if this is also a connected component, and its size. If $s$ is in a $2$-set of size up to $r + 1$, that is not a $1$-set, the first depth-first search will have a bridge connecting this $2$-set to the rest of the graph in its search tree. Thus, after the second search we know if $s$ is in a $2$-set of size up to $\mathfrak{r}$.
If we are in such a set, we do a third search to find out if this $2$-set is $2$-edge-connected and thus a $2$-class-leaf.
For every node it visits (except for the last one), the first search will issue queries to find more nodes. At each node, at most one query is made that finds a new node or tells us that there are no more neighbours. At each node, we could also make queries that give us edges to nodes that we already visited. The number of these queries is bounded by the nodes already visited. Thus, in the first search there at most $\sum_{i = 0}^r (1 + (i - 1)) = \frac{\mathfrak{r}(\mathfrak{r} + 1)}{2}$ queries. The second search is similar, but it will have to avoid certain edges, resulting in some additional queries; the number of these edges is $\mathfrak{r}$, we thus get a bound of $\frac{\mathfrak{r}(\mathfrak{r} + 1)}{2} + \mathfrak{r}$ queries. The third search also needs to avoid edges to nodes not found by the second search. But since the second search resulted in a $2$-set, there is only one such edge. The number of queries in the second search is thus bounded by $\sum_{i = 0}^{\mathfrak{r} - 1}(1 + (i - 1)) + (\mathfrak{r} - 1) + 1 = \frac{\mathfrak{r}(\mathfrak{r} - 1)}{2} + \mathfrak{r}$. This gives us a bound of at most $\frac{3\mathfrak{r}(\mathfrak{r} + 1)}{2}$ queries made by the algorithm in Figure~\ref{12leaf}.

\begin{figure}[!h]
\lstset{mathescape=true}
\begin{lstlisting}
int $\times$ int 1_2_leaf($s$, $\mathfrak{r}$)
{
  1 Do first depth-first search from $s$ for up to $\mathfrak{r} + 1$ nodes.
    Let $n_1$ be the number of nodes found.

  2 Do second depth-first search from $s$ for up to $\mathfrak{r} + 1$ nodes,
    never traversing an edge of the search tree of 1 in the same direction.
    Let $n_2$ be the number of nodes found.

  if ($n_1$ == $n_2$ == $\mathfrak{r} + 1$) // Not in a 2-set of size at most $\mathfrak{r}$
    return $(0, 0)$;

  3 Do third depth-first search from $s$ for up to $\mathfrak{r}$ nodes,
    only considering nodes found in 2,
    never traversing an edge of the search tree of 2 in the same direction.
    Let $n_3$ be the number of nodes found.

  if ($n_2$ == $n_3$ == $n_1$) // In a 2-class-leaf of size $n_3$ that is a 1-class-leaf
    return $(n_3, 0)$; 
  else if ($n_2$ == $n_3$) // In a 2-class-leaf of size $n_3$ that is not a 1-class-leaf
    return $(0, n_3)$; 
  else // In a $2$-set of size $n_2$, that is not a $2$-class-leaf
    return $(0, 0)$; 
}
\end{lstlisting}
\caption{\label{12leaf}1\_2\_leaf in pseudocode (implementation: zshg\_component2() in zshg2\_c.c)}
\end{figure}

\begin{figure}[!h]
\lstset{mathescape=true}
\begin{lstlisting}
float 1_2_leaves($\delta$, $p$, $n$)
{
  $r \coloneqq \frac{7}{2(1 - p)\delta^2}$;
  $a_1 \coloneqq a_2 \coloneqq 0$;

  $\ell \coloneqq r(3\log_e(\frac{2}{\delta}) + 6)$;

  for($i \coloneqq 0$; $i < r$; $i \coloneqq i + 1$)
  {
    $s \coloneqq$ rand_index($n$);
    $x \coloneqq$ rand_range($\frac{2}{\delta}$);
    $(b_1, b_2) \coloneqq$ 1_2_leaf($s$, $x$, &$\ell$);
    if($b_1 \neq 0$)
      $a_1 \coloneqq a_1 + b_1$;
    if($b_2 \neq 0$)
      $a_2 \coloneqq a_2 + b_2$;
  }

  return$\left(\frac{a_1n}{r}, \frac{a_2n}{r}\right)$;
}
\end{lstlisting}
\caption{\label{12leaves}Counting algorithm for $2$-edge-connectivity in C-like pseudocode (implementation: zshg2\_component() in zshg2\_c.c)}
\end{figure}

\begin{lemma}\label{lem:12leavescorrectness}
For $\delta \leq 1$, with probability at least $p$, from the return value of the algorithm in Figure~\ref{12leaves}, we get the minimum number of edge modifications necessary to make the input graph $2$-edge-connected up to an error of at most $\frac{1}{2} + \delta n$.
\end{lemma}

\begin{proof}
Let $c_1^*$ be the number of 2-class-leaves that are 1-class-leaves in $G$ and of size at most $\frac{2}{\delta}$. Let $c_2^*$ be the number of 2-class-leaves that are not 1-class-leaves in $G$ and of size at most $\frac{2}{\delta}$.
The algorithm in Figure~\ref{12leaves} gives us an estimate $(\hat{C}_1, \hat{C}_2)$ for ($c_1^*, c_2^*)$ that we will later use to estimate the distance to $2$-edge-connectivity. Let $(\hat{c}_1, \hat{c}_2)$ be the return value of the algorithm in Figure~\ref{12leaves} when ignoring the query bound $\ell$.
Let $B_{1,i}$ be the value of $b_1$ in iteration $i$, let $B_{2,i}$ be the value of $b_2$ in iteration $i$. Let $C_{2,s}$ be the size of the $2$-class-leaf that contains $s$. Let $\hat{c} \coloneqq \frac{\hat{c}_2}{2} + \frac{1}{2} + \hat{c}_1 + \frac{\delta n}{2}$ be our estimate of the distance.
Similarly to the calculation for the distance to connectivity estimate, we get:

\begin{align*}
E(B_{1,i}) &= \frac{c_1^*}{n} \leq 1. \\
E(B_{2,i}) &= \frac{c_2^*}{n} \leq 1. \\
Var(B_{1,i}) &\leq 1. \\
Var(B_{2,i}) &\leq 1. \\
Cov(B_{1,i}, B_{2,i}) &\leq 1. \\
\hat{c_1} &= \frac{n}{r}\sum_{i = 0}^{r - 1}B_{1, i}. \\
\hat{c_2} &= \frac{n}{r}\sum_{i = 0}^{r - 1}B_{2, i}. \\
E(\hat{c}_1) &= r\frac{nE(B_{1,i})}{r} = c_1^*. \\
E(\hat{c}_2) &= r\frac{nE(B_{2,i})}{r} = c_2^*. \\
Var(\hat{c_1}) &= \frac{n}{r^2}Var(B_{1, i}) \leq \frac{n^2}{r}. \\
Var(\hat{c}_2) &= \frac{n}{r^2}Var(B_{2, i}) \leq \frac{n^2}{r}. \\
Cov(\hat{c}_1, \hat{c}_2) &\leq \sqrt{Var(\hat{c_1})}\sqrt{Var(\hat{c_2})} \leq \frac{n^2}{r}. \\
E(\hat{c}) &= \frac{c_2^*}{2} + \frac{1}{2} + c_1^* + \frac{\delta n}{4}. \\
Var(\hat{c}) &= \frac{1}{4}Var(\hat{c}_2) + Var(\hat{c}_1) + \frac{1}{2} Cov(\hat{c}_1, \hat{c}_2) \leq \frac{7}{4} \frac{n^2}{r}. \\
\end{align*}

The number of 2-class-leaves bigger than $\frac{2}{\delta}$ is at most $\frac{\delta n}{2}$. Ignoring the query bound for now, we bound the probability of the estimate being far off~\cite{TchEbychef1867}:

\begin{align*}
\left\lceil\frac{c_2}{2}\right\rceil + c_1 - \frac{1}{2} - \frac{\delta n}{4} \leq E(\hat{c}) \leq \left\lceil\frac{c_2}{2}\right\rceil + c_1 + \frac{1}{2} + \frac{\delta n}{4}. \\
P\left(\left|\hat{c} - E(\hat{c})\right| > \frac{3\delta n}{4}\right) < \frac{Var(\hat{c})}{\left(\frac{3\delta n}{4}\right)^2} \leq \frac{28}{9r\delta} = \frac{8}{9}(1 - p). \\
\end{align*}

Let $Q_i$ be the number of queries made by a call to 1\_2\_leaves in the algorithm in Figure~\ref{12leaves}.

\begin{align*}
E(Q_i) &\leq \sum_{j = 1}^{\frac{2}{\delta} - 1} \left(\frac{1}{j} - \frac{1}{(j + 1)^2}\right) \frac{3j(j + 1)}{2} + \frac{1}{\left(\frac{2}{\delta}\right)^2}\frac{3\frac{2}{\delta}(\frac{2}{\delta} + 1)}{2} = \\
&= \sum_{j = 1}^{\frac{2}{\delta} - 1} \frac{6j + 3}{3j(j + 1)} + \frac{6 + 3\delta}{16} < 3\log_e\left(\frac{2}{\delta}\right).\\
E(Q_i^2) &\leq \sum_{j = 1}^{\frac{2}{\delta} - 1}\left(\frac{1}{j^2} - \frac{1}{(j + 1)^2}\right)\left(\frac{3j(j + 1)}{2}\right)^2 + \frac{1}{\left(\frac{2}{\delta}\right)^2}\left(\frac{3\frac{2}{\delta}(\frac{2}{\delta} + 1)}{2}\right)^2 = \\
&= \sum_{j = 1}^{\delta{2}{\delta} - 1}\left(\frac{9(j + 1)^2}{4} - \frac{9j^2}{4}\right) + \frac{9(\frac{2}{\delta} + 1)^2}{4} = \frac{9}{\delta^2} + \frac{9}{2\delta}. \\
Var(Q_i) &= E(Q_i^2) - E(Q_i)^2 \leq \frac{9}{\delta^2} + \frac{9}{2\delta}. \\
\end{align*}

Which we use to bound the probability of hitting the query bound~\cite{TchEbychef1867}.

\begin{align*}
P\left(\sum_{i = 0}^{r - 1}Q_i > \ell\right) \leq P\left(\sum_{i = 0}^{r - 1}Q_i > r (E(Q_0) + 6)\right) < \frac{Var(Q_0)}{36r}< \frac{1}{9}(1 - p). \\
P\left(\left|\hat{C} - E(\hat{c})\right| > \frac{3\delta n}{4}\right) \leq P\left(\left|\hat{c} - E(\hat{c})\right| > \frac{3\delta n}{4}\right) + P\left(\sum_{i = 0}^{r - 1}Q_i > \ell\right) < 1 - p. \\
\end{align*}
\end{proof}

\begin{lemma}\label{lem:12leavescomplexity}
The algorithm in Figure~\ref{12leaves} has query complexity $O(\frac{1}{(1 - p)\delta^2} \log_e(\frac{1}{\delta}))$.
\end{lemma}

\begin{proof}
This follows directly from the choice of $\ell$.
\end{proof}

For general $k$, we use a generalization of Lemma~\ref{Lemma:Eswaran1976} above.

The \emph{edge-demand} of a set of nodes $U \subsetneq V$ is

\begin{displaymath}
\Phi_k(U) \coloneqq \max \left\{0, k - d(U), \sum_{W \sqsubset U}\Phi_k(W) \right\}, \\
\Phi_k(V) \coloneqq \sum_{W \sqsubset V}\Phi_k(W). \\
\end{displaymath}

\begin{lemma}[Naor et alii~\cite{Naor1997} for connected graphs, generalized by Marko~\cite{Marko2005}]\label{Lemma:Naor1997}
Let $G = (V, E)$ be a graph and $k \in \mathbb{N}$. The distance of $G$ to $k$-edge-connectivity is $\frac{1}{|E|}\left\lceil \frac{\Phi_k(V)}{2} \right\rceil$.
\end{lemma}

To estimate $\Phi_k(V)$ we can estimate the number of extreme sets $U$, for which $\Phi_k(U) \neq \sum_{W \sqsubset U}\Phi(W)$, and their $\Phi_k(U)$ and then get the estimate of $\Phi_k(V)$ as the weighted sum.

For $k = 3$, there are $4$ kinds of such $U$: 3-class-leaves that are connected components (for these, $\Phi_k = 3$), 3-class-leaves that are connected to the rest of the graph by a single bridge (for these, $\Phi_k = 2$), 3-class-leaves that are connected to the rest of the graph by a exactly two edges and are not inside the fourth kind (for these, $\Phi_k = 1$), connected components that are 2-edge-connected and contain exactly 2 3-edge-leaves (for these, $\Phi_k = 3$). Sets of the fourth kind contribute $3$ each to $\Phi_3$, but for the algorithm it is simpler to make them contribute $1$ and instead ignore the condition on sets of the third kind not being in the fourth kind.
Let $C_0$ be the number of of 3-class-leaves that are connected components, let $C_1$ be the number of of 3-class-leaves that are connected to the rest of the graph by a single bridge, let $C_2$ be the number of $3$-class-leaves that are connected to the rest of the graph by a exactly two edges, let $C_3$ be the number of connected components that are 2-edge-connected and contain exactly 2 3-edge-leaves.

\begin{displaymath}
\Phi_3(V) = 3C_0 + 2C_1 + 1C_2 + 1C_3 - k.
\end{displaymath}

To find out if a node $s$ is in any such set of size at most $\mathfrak{r}$, we use the algorithm in Figure~\ref{123leaf}. It returns the size of the containing sets: Figure~\ref{123leafcases} illustrates the possible cases: If $s$ is in a $3$-set larger than $\mathfrak{r}$, the algorithm returns $(0, 0, 0, 0)$ to contribute nothing to any estimate of $C_0, \ldots$. If $s$ is in a $3$-set $m$ of size at most $\mathfrak{r}$ that is also a connected component, the algorithm returns $(|m|, 0, 0, 0)$ to contribute to the estimate of $C_0$ only. If $s$ is in a $3$-set $m$ of size at most $\mathfrak{r}$ that is connected to the rest of the graph by a bridge, the algorithm returns $(0, |m|, 0, 0)$ to contribute to the estimate of $C_1$ only. If $s$ is in a $3$-set $m$ that is connected to another $3$-set $n$ by two edges, with $m \cup n$ being a connected component of size at most $\mathfrak{r}$, the algorithm returns $(0, 0, |m|, |m + n|)$ to contribute to the estimates of both $C_2$ and $C_3$. Otherwise, $s$ is in a $3$-set $m$ of size at most $\mathfrak{r}$ that is connected to the rest of the graph by two edges, the algorithm returns $(0, 0, |m|, 0)$ to contribute to the estimate of $C_3$ only.
To decide if a node is in a set of the fourth kind, the algorithm uses the plant graph~\cite{Dinitz1976} of the connected component, which can be computed in time quadratic in the number of nodes in the connected component~\cite{Karzanov1986}. This results in total query complexity $O(\mathfrak{r}^3)$

\begin{figure}[!h]
\centering
\includegraphics[scale=0.5]{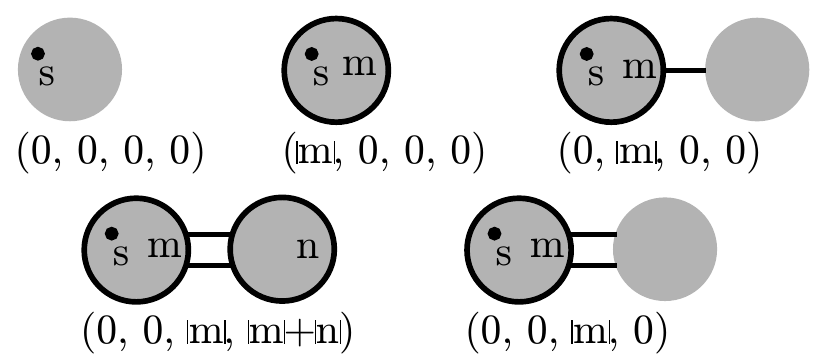}
\caption{\label{123leafcases}Cases for the algorithm in Figure~\ref{123leaf}}
\end{figure}

We can use this algorithm to construct an algorithm for estimating the distance to 3-edge-connectivity similar to how we did so for 2-edge-connectivity above.

\begin{figure}[!h]
\lstset{mathescape=true,basicstyle=\small\ttfamily}
\begin{lstlisting}
int $\times$ int $\times$ int $\times$ int 1_2_3_leaf($s$, $\mathfrak{r}$)
{
  0 Do a depth-first search from $s$ for up to $\mathfrak{r} + 1$ nodes.
    Let $n_0$ be the number of nodes found. Let $m_0$ be the set of nodes found.

  1 Do a depth-first search from $s$ for up to $\mathfrak{r} + 1$ nodes,
    never traversing an edge of the search tree of 1 in the same direction.
    Let $n_1$ be the number of nodes found.

  if($n_0 \leq \mathfrak{r}$ && $n_0$ == $n_1$) // $s$ is in a 2-connected connected component of size at most $\mathfrak{r}$
    {
      Compute the plant graph $\Gamma$ of the subgraph induced by $m_0$.
      if ($\Gamma$ has exactly 2 nodes of degree 1 and no cycles)
        $n_2 \coloneqq n_0$;
      else
        $n_2 \coloneqq 0$;
    }
  else
    $n_2 \coloneqq 0$;

  2 For every edge traversed in the first depth-first search,
    invoke 2_set searching for up to $\mathfrak{r} + 1$ nodes in the graph with that edge omitted.
    Let $m'$ be a smallest set found by 2_set.

  if ($|m|$ > $\mathfrak{r}$) // Not in a 3-set of size at most $\mathfrak{r}$
    return $(0, 0, 0, n_2)$;

  3 Do a depth-first search from $s$ for up to $|m'|$ nodes in the subgraph induced by $m'$.

  4 Do depth-first search from $s$ for up to $|m'|$ nodes in the subgraph induced by $m'$,
    never traversing an edge of the search tree of 3 in the same direction.

  5 For every edge traversed in the previous depth-first search,
    invoke 2_set searching for up to $|m'|$ nodes in the subgraph induced by $m'$ with that edge omitted.
    Let $m$ be a smallest set found by 2_set.

  if ($|m| \neq |m'|$) // In a 3-set of size $|m|$ that is not a 3-class-leaf
    return $(0, 0, 0, n_2)$; 

  if($d(m)$ == 0) // In a connected component that is a 3-class leaf.
    return $(|m|, 0, 0, n_2)$;
  else if ($d(m)$ == 1) // In a 3-class leaf connected to the rest by a single bridge
    return $(0, |m|, 0, n_2)$;
  else // $d(m)$ == 2 // In a 3-class leaf connected to the rest by two edges
    return $(0, 0, |m|, n_2)$;
}
\end{lstlisting}
\caption{\label{123leaf}1\_2\_3\_leaf in pseudocode}
\end{figure}

From our estimate of the number of components, we also obtain an estimate for the distance to eulerianity (Figure~\ref{eulerianitydistancealgorithm}). The algorithm \texttt{1\_sets\_even} is a variant of \texttt{1\_sets} that only counts components in which all nodes have even degree.

\begin{figure}[!h]
\lstset{mathescape=true}
\begin{lstlisting}
float euler_distance($\epsilon$, $p$, $n$)
{
  	$\delta \coloneqq \frac{\epsilon}{d}$;
	$c \coloneqq$ 1_sets_even$\left(\frac{\delta}{2}, \frac{1 - p}{2}\right)$;
	$m \coloneqq \frac{2}{\delta^2(1 + p)}$;
	$a \coloneqq 0$;

	for($i \coloneqq 0$; $i < m$; $i \coloneqq i + 1$)
		if($d(s)$ odd)
			$a \coloneqq a + 1$;
	$u \coloneqq \frac{an}{m}$;
	$e \coloneqq c + \frac{u}{2}$;

	return$\left(\frac{e}{nd}\right)$;
}
\end{lstlisting}
\caption{\label{eulerianitydistancealgorithm}Algorithm for estimating the distance to eulerianity in C-like pseudocode (the implementation resides in zshg\_euler\_distance() in zshg\_c.c)}
\end{figure}

\begin{lemma}\label{lem:eulerianitydistancecorrectness}
With probability at least $p$, value returned by the algorithm in Figure~\ref{eulerianitydistancealgorithm} is the distance of the graph to eulerianity with an error of at most $\epsilon$.
\end{lemma}

\begin{proof}
A graph is eulerian if it is connected and all nodes have even degree~\cite{Euler1741,Listing1847,Hierholzer1873}.

Let $C$ be the number of connected components, in which all nodes have even degree. Let $U$ be the number of nodes of odd degree. Then the minimum number of edge modifications necessary to make the graph connected is $C + \frac{U}{2}$: In every component, the number of nodes of odd degree is even~\cite{Koenig1935}, so one can add a cycle that goes through all components (using preexisting edges within components that contain nodes of odd degree), requiring only that number of edges. It is not possible to do with fewer edge modifications: At every component that contains only nodes of even degree, we have to add at least 2 incident edges to a node, and at every node of odd degree in the graph, we have to add or remove an edge.
Let q be the probability of the algorithm returning a value outside the error bound.

\begin{align*}
E(c) = C, E(a_i) = \frac{U}{n}, E(a_i^2) = \frac{U}{n}, Var(u) \leq \frac{U}{n}. \\
E(u) = U, Var(u) = \frac{n^2}{m} Var(a_i) \leq \frac{n}{m}U \leq \frac{n^2}{m}. \\
P(|u - U| > \epsilon dn) < \frac{Var(u)}{\epsilon^2 d^2n^2} \leq \frac{1}{m\epsilon^2d^2} = \frac{1 - p}{2}. \\
q \leq P(|c - C| > \frac{\epsilon dn}{2}) + P(|u - U| > \epsilon dn) \leq 1 - \frac{1 - p}{2} + 1 - \frac{1 + p}{2} = 1 - p. \\
\end{align*}
\end{proof}

\begin{lemma}\label{lem:eulerianitydistancecomplexity}
The algorithm in Figure~\ref{eulerianitydistancealgorithm} has query complexity $O\left(\frac{1}{(1 - p)\delta^2} \log_e(\frac{1}{\delta})\right)$.
\end{lemma}

\begin{proof}
For \texttt{1\_sets\_even} we get that complexity like in Lemma~\ref{lem:componentcomplexity} and the number of further queries is $m \in O\left(\frac{1}{\delta^2(1 - p)}\right) \subseteq O\left(\frac{1}{(1 - p)\delta^2} \log_e(\frac{1}{\delta})\right)$.
\end{proof}

\begin{theorem}
There are estimates for the distance to connectivity, $2$-edge-connectivity and eulerianity with complexity $O(\frac{1}{(1 - p)\delta^2} \log_e(\frac{1}{\delta}))$.
There is an estimate for the distance to $3$-edge-connectivity with complexity $O(\frac{1}{(1 - p)\delta^3} \log_e(\frac{1}{\delta}))$.
\end{theorem}

\begin{proof}
For the distance to connectivity this follows from Lemmata \ref{lem:componentcorrectness}, \ref{lem:componentcomplexity}.
For the distance to $2$-edge-connectivity, this follows from Lemmata \ref{Lemma:Eswaran1976}, \ref{lem:12leavescorrectness}, \ref{lem:12leavescomplexity}.
For the distance to eulerianity, this follows from Lemmata \ref{lem:eulerianitydistancecorrectness}, \ref{lem:eulerianitydistancecomplexity}.
For the distance to $3$-edge-connectivity, this can be proven using Lemma \ref{Lemma:Naor1997} and the algorithm in Figure~\ref{123leaf}.
\end{proof}

\section{Tolerant Connectivity Test}\label{Tolerant}

\begin{figure}[!h]
\lstset{mathescape=true}
\begin{lstlisting}
bool zshg_tolerant($\epsilon_1$, $\epsilon_2$, $p$, $d$, $n$)
{
  $\delta \coloneqq \frac{\epsilon_2 - \epsilon_1}{4}d$;
  $e \coloneqq$ 1_sets($\delta$, $p$, $n$) - 1;
  return$\left(e \leq \frac{(\epsilon_1 + \epsilon_2)dn}{4}\right)$;
}
\end{lstlisting}
\caption{\label{Tolerantalgorithm}Tolerant connectivity testing algorithm in C-like pseudocode (the implementation resides in zshg\_tolerant() in zshg\_c.c)}
\end{figure}

\begin{lemma}\label{lem:tolerantcorrectness}
For any graph that is $\epsilon_1$-close to connectivity, the algorithm in Figure~\ref{Tolerantalgorithm} returns true with probability at least $p$.
For any graph that is $\epsilon_2$-far from connectivity, the algorithm in Figure~\ref{Tolerantalgorithm} returns false with probability at least $p$.
\end{lemma}

\begin{proof}
Let $m = \frac{dn}{2}$ be the number of edges in the graph.

Case 1: The graph is $\epsilon_1$-close to connectivity, i.\,e. it consists of at most $\epsilon_1m + 1$ connected components.

\begin{gather*}
P\left(e > \frac{(\epsilon_1 + \epsilon_2)dn}{4}\right) = P\left(e > \frac{\epsilon_1 + \epsilon_2}{2}m\right) \leq\\
P\left(\textrm{1\_sets}(\delta, p, n) > \frac{\epsilon_1 + \epsilon_2}{2}m + 1\right) = P\left(\textrm{1\_sets}(\delta, p, n) > (\epsilon_1m + 1) + \delta n\right) \leq\\
(1 - p).
\end{gather*}

Case 2: The graph is $\epsilon_2$-far from connectivity, i.\,e. it consists of at least $\epsilon_2m + 2$ connected components.

\begin{gather*}
P\left(e \leq \frac{(\epsilon_1 + \epsilon_2)d}{4}\right) = P\left(e \leq \frac{\epsilon_1 + \epsilon_2}{2}m\right) \leq\\
P\left(\textrm{1\_sets}(\delta, p, n) \leq \frac{\epsilon_1 + \epsilon_2}{2}m + 1\right) = P\left(\textrm{1\_sets}(\delta, p, n) \leq (\epsilon_2m + 1) - \delta n\right) \leq\\
(1 - p).
\end{gather*}

Case 3: The graph is neither $\epsilon_1$-close to nor $\epsilon_2$-far from connectivity. We don't care.
\end{proof}

In a similar way, we can obtain tolerant testers for $2$-edge-connectivity, $3$-edge-connectivity and eulerianity.

\begin{theorem}
There are tolerant testers for connectivity, $2$-edge-connectivity and eulerianity with complexity $O\left(\frac{1}{(1 - p)(\epsilon_2 - \epsilon_1)^2 d^2}\log_e(\frac{1}{(\epsilon_2 - \epsilon_1) d})\right)$.
There is a tolerant tester for $3$-edge-connectivity with complexity $O\left(\frac{1}{(1 - p)(\epsilon_2 - \epsilon_1)^3 d^3}\log_e(\frac{1}{(\epsilon_2 - \epsilon_1) d})\right)$.
\end{theorem}

\begin{proof}
For connectivity this follows from Lemmata \ref{lem:componentcorrectness},\ref{lem:componentcomplexity},\ref{lem:tolerantcorrectness}. The proofs for the others are similar.
\end{proof}

\section{Remarks}\label{Remark}

While the results on correctness also hold for multigraphs, the results on query complexity don't.

All of our algorithms can be parallelized easily and we did so in the implementations. The parallel versions offer an advantage when multiple pending queries can be answered more efficiently, e.g. in the case of large (too big to fit into RAM) graphs stored on an SSD (current SSDs typically achieve maximum throughput for random reads at about 16 simultaneous pending reads) or in the case of the queries being processed by a remote server on a network.

\bibliographystyle{plain}
\bibliography{zshg}

\end{document}